\documentclass[12pt,fleqn]{article}
\usepackage{indentfirst}
\usepackage{amsfonts, color}
\usepackage{amsmath}
\usepackage{amssymb}
\usepackage{amstext}
\usepackage{graphicx,epsfig}
\usepackage{hyperref}
\usepackage{amsthm}
\usepackage{latexsym}
\newtheorem{thm}{Theorem}
\newtheorem{lemma}{Lemma}

\newcommand{\tr}{\mathrm{tr}}

\begin{document}

\begin{center}
\raggedright\LARGE\bf Multipartite separability of density matrices of graphs 
\end{center}

\begin{center}
 \raggedright Hui Zhao$^1$,  Jing-Yun Zhao$^1$, and Naihuan Jing $^{2,3 *}$
 \end{center}

\footnotetext{\hspace*{-.45cm}\footnotesize $^{*}$ Corresponding author. E-mail: jing@ncsu.edu \\
$^1$ College of Applied Sciences, Beijing University of Technology, Beijing 100124, China\\
\small $^2$ Department of Mathematics, North Carolina State University, Raleigh, NC 27695, USA \\
\small $^3$ Department of Mathematics, Shanghai University, Shanghai 200444, China
}

\vspace*{2mm}
\begin{center}
\begin{minipage}{17cm}
\parindent 20pt\footnotesize
\noindent {\bf Abstract} A new layers method is presented for multipartite separability of density matrices from simple graphs. Full separability of tripartite states is studied for graphs on degree symmetric premise. The models are generalized to multipartite systems by presenting a class of fully separable states arising from partially symmetric graphs.\\
\end{minipage}
\end{center}
\begin{center}
\begin{minipage}{17cm}
\begin{minipage}[t]{2.3cm}{\bf Keywords:}\end{minipage}
\begin{minipage}[t]{17cm}
\noindent Laplacian Matrices, Partially symmetric, Degree symmetric, Separability
\end{minipage}\par\vglue8pt
{\bf PACS: }\noindent 03.65.Ud, 02.10.Ox, 03.67.Mn
\end{minipage}
\end{center}

\section{Introduction}

 Quantum entanglement is one of the most fascinating features of quantum theory and has numerous applications in quantum information processing, secure communication and channel protocols [1,2,3]. The theory of graphs, a well-developed mathematical area, has been intensively used
  in network systems, optimization, and other fields [4].

In [5] the notion of the density matrix of a graph was introduced and it was shown that mixed states correspond to the graphical property of
uniform mixture, and the Laplacian matrices have been studied in terms of entanglement properties such as von Neumann entropy and concurrence.
Corresponding to simple graphs, the graph states are introduced as a family of multipartite quantum states [6], and their nice
entanglement structures have been extensively used in quantum computations. For instance, new algorithm based on
graph states [7] was given and showed improvement in comparison with exploiting the physics of optically active multi-level nano structures [8].
Graph theoretic methods have also been developed to analyze maximally entangled pure states distributed between a number of different parties [9].
Recently, theoretical principle of representing the quantum state and local unitary graph was established in [10]. Conditions for separability of generalized Laplacian matrices of weighted graphs with unit trace were given in [11]. Further results on the multipartite separability of Laplacian matrices of graphs were provided in [12,13]. In [14] the authors characterized the set of graphs whose separability are invariant under graph isomorphisms. Two classes of generalized graph product states were also constructed in [15].
These have provided an alternative interesting graph theoretic approach to separability and several well-known criteria have been formulated in the new method. For instance, it was proved that the degree criterion is equivalent to the PPT-criterion [16]. And a degree condition to test separability of density matrices of graphs was described in [17]. It was further shown that the well-known matrix realignment criterion can be used to test separability for a class of quantum states (cf. [18]).

On the other hand, Dutta et. al. [19] introduced the concept of partially symmetric graphs and degree symmetric graphs for bipartite quantum states. They presented some outstanding results for bipartite quantum states. As mentioned in [19] the simple assignment of direction to particles as ``vertical'' and ``horizontal'' will no longer be possible for three or more particles using the original layers method. We will develop a new method to solve this problem and generalize the separable results to multipartite systems.

The current work aims to study multipartite separability of density matrices from simple graphs. Using the graph-theoretic method, we are able to distinguish entangled and separable Laplacian matrices as well as generalizing the results to mulipartite systems. The paper is organized roughly as follows. Basic information on graphs is given in Section 2. In Section 3, we study separability of tripartite states and multipartite states defined by simple graphs by introducing layers in graph $G$ and the properties of degree symmetric graphs. In Section 4, the relationship between partially symmetric and degree symmetric is given. We then prove that a class of partially symmetric graphs is fully separable using the new layer method. Conclusions are given in Section 5.\\

\section{The Laplacian Matrices of Graphs} \label{The bound}
We begin by recalling some basic notions for graphs. Let $G=(V(G),E(G))$ be a graph with vertex set $V(G)$
and edge set
$E(G)\subset\{(i,j):i, j\in V(G)\}$. A loop is an edge of the form $(i,i)$. We are only concerned with simple graphs, i.e, graphs without loops and multiple edges. Suppose $G$ has $n$ vertices, i.e. $|V(G)|=n$. The  adjacency matrix $A(G)$ of the graph $G$ is the $n\times n$ matrix with $(i,j)$-th entry defined by
\begin{equation}
[A(G)]_{i,j}=\left\{
             \begin{array}{lr}
             1,\ \mathrm{if}\ (i,j)\in E(G);  \\
             0, \ \mathrm{otherwise}.
             \end{array}
\right.
\label{A}
\end{equation}
The degree $d_G(v_i)$ of vertex $v_i \in V(G)$ is the number of edges adjacent to $v_i$. The degree matrix $D(G)$ is the diagonal $n\times n$-matrix with diagonal entries $d_G(v_i)$.

The Laplacian matrix $L(G)$ and normalized Laplacian matrix $Q(G)$ of the graph $G$ are defined respectively by
\begin{align}
L(G)=D(G)-A(G), \qquad Q(G)=D(G)+A(G).
\end{align}

Recall that the density matrix of a finite dimensional quantum mechanical system $H_n$ is a Hermitian operator $\rho$ acting on $H_n$ that is positive semidefinite with unit trace. Its combinatorial counterpart is defined as follows.

{\bf \noindent Definition 1}.
The density matrix of graph $G$ is defined as the matrix
\begin{align}
\rho_l(G)=\frac{L(G)}{\tr(L(G))}, \qquad \mbox{or}\qquad \rho_q(G)=\frac{Q(G)}{\tr(Q(G))},
\end{align}
where $L(G)$ (resp. $Q(G)$) are the combinatorial (resp. normalized) Laplacian matrix of the graph $G$.

We need to recall the concept of separability to present our results.

{\bf \noindent Definition 2}[20].
A bipartition of the system $S$=$\{1,2,\cdots,n $\} is a pair $(A,\bar{A})$, with $ 1\leq n_A\leq n_{\bar{A}}$, where $A\subset S$, $\bar{A}=S\setminus A\ (i.e.\ S=A+\bar{A})$ and $n_A=|A|$, the cardinality of $A$.

{\bf \noindent Definition 3}.
A tripartite state is separable with respect to the bipartition $A|BC$ if it can be written as
\begin{eqnarray}
\rho_{A\mid BC}=\sum_kq_k|\phi^k_A\rangle\langle\phi^k_A|\otimes|\varphi^k_{BC}\rangle\langle\varphi^k_{BC}|,
\label{B}
\end{eqnarray}
where $q_k$ is a probability distribution, i.e. $\sum_kq_k=1$ and $q_k\geq0$ for all $k$.

Other bipartitions $B|AC$, $C|AB$ are defined similarly, and the notion can be generalized to the multipartite quantum systems.

{\bf \noindent Definition 4}. A combinatorial state $\rho$ is called biseparable if it can be written as
\begin{eqnarray}
\rho=\sum p_{A|\bar{A}}\ \rho^{sep}_{A|\bar{A}},
\label{C}
\end{eqnarray}
where the sum runs through all possible bipartitions $A|\bar{A}$ and $p_{A|\bar{A}}$ is
a probability distribution. Here $\bar{A}$ is the complement of the
subset $A$ of the vertex set $\{1,2,\cdots,n $\}.

{\bf \noindent Definition 5}.
A density matrix $\rho$ is fully separable in $H_1\otimes H_2\otimes\cdots\otimes H_n$ if it can be written as
\begin{eqnarray}
\rho=\sum_i q_i\ \rho^1_i\otimes\rho^2_i\cdots\otimes\rho^n_i,
\label{D}
\end{eqnarray}
where $q_i$ is a probability distribution and
$\rho^j_i$ are density matrices in the sbusystem $H_j(j=1,2,\cdots,n)$.

\section{Degree Symmetric Graphs and Separability }

In this section, we study separability of tripartite states of dimension $m\times n\times q$ which arise from simple graphs of $mnq$ vertices acting on $H_A\otimes H_B\otimes H_C$. These results are then generalized to multipartite quantum systems.

Let us start with the notion of {\it layers} in graph $G$. Suppose the vertex set $V(G)$ of the graph $G$ has $mnq$ vertices labelled by
integers $1,2,\cdots,mnq$. We partition $V(G)$ into $m$ subsets $C_1,C_2,\cdots,C_m$ called layers such that each layer consists of $nq$ vertices.
Write
$$C_i=\{v_{i,1,1},\cdots,v_{i,1,q},\cdots, v_{i,n,1},\cdots,v_{i,n,q}\}.$$
Thus the matrix $A(G)$ is partitioned into blocks as follows.
  \begin{eqnarray}
A(G)=
\left[ \begin{array}{ccccc}
           A_1,_1&A_1,_2&\cdots&A_1,_{m-1}&A_1,_m   \\
           A_2,_1&A_2,_2&\cdots&A_2,_{m-1}&A_2,_m     \\
           \vdots&\vdots&\ddots&\vdots&\vdots   \\
           A_{m-1},_1&A_{m-1},_2&\cdots&A_{m-1},_{m-1}&A_{m-1},_m     \\
           A_m,_1&A_m,_2&\cdots&A_m,_{m-1}&A_m,_m    \\
           \end{array}
      \right ],
\label{F}
\end{eqnarray}
  where $A_{i,k}\ (i,k=1,2,\cdots,m)$ are matrices of size $nq$ representing edges between $C_i$ and $C_k$. Each layer $C_i$ can be partitioned
  further into $n$ layers $C_{i,1}, C_{i,2},\cdots, C_{i,j},\cdots, C_{i,n}$
  with $q$ vertices each. Write $C_{i,j}$=$\{v_{i,j,1}, v_{i,j,2},\cdots,v_{i,j,q} $\} where\ $v_{i,j,k}=nq(i-1)+q(j-1)+k$. Therefore
  $A_{i,k}$ is written as a block matrix as follows.
  \begin{eqnarray}
A_{i,k}=
\left[ \begin{array}{ccccc}
           A_{i1},_{k1}&A_{i1},_{k2}&\cdots&A_{i1},_{kn}   \\
           A_{i2},_{k1}&A_{i2},_{k2}&\cdots&A_{i2},_{kn}   \\
           \vdots&\vdots&\ddots&\vdots&   \\
           A_{in},_{k1}&A_{in},_{k2}&\cdots&A_{in},_{kn}    \\
           \end{array}
      \right ],
\label{G}
\end{eqnarray}
  where $A_{ij,kl}$ are the (sub)-adjacency matrices of size $q$ representing edges between $C_{i,j}$ and $C_{k,l}$. As the
   adjacency matrix is symmetric, we have that
   $A_{ij,kl}^T$=$A_{kl,ij}$ $(i\neq k,j\neq l)$ and $A_{ij,ij}$=$A_{ij,ij}^T$. Putting these together, $A(G)$ is written into the following form
\begin{eqnarray}
A(G)=
\left[ \begin{array}{cccccccc}
           A_{11},_{11}&A_{11},_{12}&\cdots&A_{11},_{1n}&\cdots&A_{11},_{m1}&\cdots&A_{11},_{mn}   \\
           A_{12},_{11}&A_{12},_{12}&\cdots&A_{12},_{1n}&\cdots&A_{12},_{m1}&\cdots&A_{12},_{mn}   \\
           \vdots&\vdots&\ddots&\vdots&\ddots&\vdots&\ddots&\vdots   \\
           A_{1n},_{11}&A_{1n},_{12}&\cdots&A_{1n},_{1n}&\cdots&A_{1n},_{m1}&\cdots&A_{1n},_{mn}   \\
            \vdots&\vdots&\ddots&\vdots&\ddots&\vdots&\ddots&\vdots   \\
           A_{m1},_{11}&A_{m1},_{12}&\cdots&A_{m1},_{1n}&\cdots&A_{m1},_{m1}&\cdots&A_{m1},_{mn}   \\
           \vdots&\vdots&\ddots&\vdots&\ddots&\vdots&\ddots&\vdots   \\
           A_{mn},_{11}&A_{mn},_{12}&\cdots&A_{mn},_{1n}&\cdots&A_{mn},_{m1}&\cdots&A_{mn},_{mn}   \\
           \end{array}
      \right ],
\label{H}
\end{eqnarray}

Partially transposed graph of a bipartitie quantum state was defined in [19]. We generalize the notion to tripartite systems.

{\bf \noindent Definition 6}.
A graph theoretical partial transpose $(GTPT)$ is an operation on the tripartite graph $G$ by
replacing all existing edges $(v_{i,j,k},v_{s,u,v}), i\neq s$ by $(v_{s,j,k},v_{i,u,v})$ and keeping all other edges unchanged.

$GTPT$ generates a new simple graph $G'=(V(G'),E(G'))$, where $V(G')=V(G)$ with the same labelling. Similarly one can
define other $GTPT$ by replacing $(v_{i,j,k},v_{s,u,v}), j\neq u$ by $(v_{i,u,k},v_{s,j,v})$ or $(v_{i,j,k},v_{s,u,v}), k\neq v$ by $(v_{i,u,v},v_{s,j,k})$. In this paper, we only consider partial transpose in the sense of Definition 6, separability criteria in the other two cases are easily obtained similar method.

Note that
\begin{equation}
A(G)^{T_A}=\left\{
             \begin{array}{lr}
             1,\ (nq(s-1)+q(j-1)+k, nq(i-1)+q(u-1)+v)\in E(G');  \\
             0, \ otherwise.
             \end{array}
\right.
\label{I}
\end{equation}
where $T_A$ denotes the partial transpose corresponding to subsystem $A$. We have $A^{T_A}(G)=A(G')$.

Next we define degree symmetric graphs as follows.

{\bf \noindent Definition 7}. A graph $G$ is called degree symmetric if $d_G(u)=d_{G'}(u)$ for $u\in V(G)=V(G')$.

\begin{thm}
Suppose the graph $G$ is degree symmetric under GTPT, then
full separability of $\rho_l(G)$ implies full separability of $\rho_l(G')$, and full separability of $\rho_q(G)$ implies full separability of $\rho_q(G')$.
\label{1}
\end{thm}
\begin{proof} Let $G'$ be the GTPT of $G$ with respect to subsystem $A$. If $G$ is degree symmetric, then $D(G')=D(G)$ and $A(G)^{T_A}=A(G')$,
where $T_A$ is the partial transpose.
Assume that $\rho_l(G)$ is fully separable, then $\rho_l(G)$ can be written as
$$\rho_l(G)=\sum_i p_i\ \rho^{A}_i\otimes \rho^{B}_i\otimes \rho^{C}_i.$$
Subequently $\rho_l(G)^{T_A}=\sum_i p_i\ (\rho^{A}_i)^{T_A}\otimes \rho^{B}_i\otimes \rho^{C}_i$ is positive semidefinite. It follows from $D(G')=D(G)$ and $A(G')=A(G)^{T_A}$ that $\rho_l(G')=\rho_l(G)^{T_A}$. Therefore $\rho_l(G')$ is fully separable.
\end{proof}

The above results can be generalized to multipartite quantum systems. Let $H_i$ be an $N_i-$dimensional complex Hilbert space, and $\rho$ a density matrix defined on $H_1\otimes H_2\otimes\cdots\otimes H_n$. The layers of graph $G$ are defined as follows.

Let $G$ be a graph of $N_1N_2\cdots N_n$ vertices. Label $V(G)$ by integers $1,2,\cdots,N_1N_2\cdots N_n$ and partition $V(G)$ into $N_1$ layers {$C_1, \cdots,C_{i_1},\cdots,C_{N_1}$}
 with $N_2N_3\cdots N_n$ vertices in each layer. Write $C_{i_1}$=$\{v_{i_1,1,\ldots,1},\cdots,v_{i_1,1,\ldots,N_n},\cdots,v_{i_1,N_2,N_3,\cdots, N_n}$\} for $1\leq i_1\leq N_1$. This actually partitions $A(G)$ into blocks as follows.
 \begin{eqnarray}
A(G)=
\left[ \begin{array}{ccccc}
           A_1,_1&A_1,_2&\cdots&A_1,_{{N_1}-1}&A_1,_{N_1}   \\
           A_2,_1&A_2,_2&\cdots&A_2,_{{N_1}-1}&A_2,_{N_1}     \\
           \vdots&\vdots&\ddots&\vdots&\vdots   \\
           A_{{N_1}-1},_1&A_{{N_1}-1},_2&\cdots&A_{{N_1}-1},_{{N_1}-1}&A_{{N_1}-1},_{N_1}     \\
           A_{N_1},_1&A_{N_1},_2&\cdots&A_{N_1},_{{N_1}-1}&A_{N_1},_{N_1}    \\
           \end{array}
      \right ],
\label{J}
\end{eqnarray}
where $A_{i_1,j_1} (i_1,j_1=1,2,\cdots,N_1)$ are matrices of size $N_2N_3\cdots N_n$ representing edges between $C_{i_1}$ and $C_{j_1}$. Next,
each layer $C_{i_1}$ can be partitioned further into $N_2$ layers
$$C_{i_1,1},\cdots,C_{i_1,i_2},\cdots,C_{i_1,N_2}$$
with $N_3N_4\cdots N_n$ elements each, where $C_{i_1,i_2}$=$\{v_{i_1,i_2,1,\ldots, 1},\cdots,v_{i_1,i_2,N_3,N_4,\cdots, N_n}$\}. This will partition $A_{i_1,j_1}$ into blocks as follows.

\begin{eqnarray}
A_{i_1,j_1}=
\left[ \begin{array}{ccccc}
           A_{i_11},_{j_11}&A_{i_11},_{j_12}&\cdots&A_{i_11},_{j_1N_2}   \\
           A_{i_12},_{j_11}&A_{i_12},_{j_12}&\cdots&A_{i_12},_{j_1N_2}  \\
           \vdots&\vdots&\ddots&\vdots&   \\
          A_{i_1N_2},_{j_11}&A_{i_1N_2},_{j_12}&\cdots&A_{i_1N_2},_{j_1N_2}    \\
           \end{array}
      \right ],
\label{K}
\end{eqnarray}
where $A_{i_1i_2,j_1j_2}(i_1,j_1=1,2,\cdots,N_1,\ i_2,j_2=1,2,\cdots,N_2\ )$
are matrices of size $N_3N_4N_5\cdots N_n$ representing edges between $C_{i_1,i_2}$ and $C_{j_1,j_2}$. In this way, we continue partitioning $G$  until each layer has $N_n$ vertices. The $(n-1)th$-layer is represented by $$C_{{i_1},{i_2},\cdots,{i_{n-1}}}=\{v_{{i_1},{i_2},{\cdots},{i_{n-1}},{1}},\ v_{{i_1},{i_2},{\cdots},{i_{n-1}},{2}},\\\cdots\cdots,
v_{{i_1},{i_2},{\cdots},{i_{n-1}},{N_n}}\}, $$
where $i_1=1,2,\cdots,N_1;\ i_2=1,2,\cdots,N_2;\ \cdots\cdots;\ i_{n-1}=1,2,\cdots,N_{n-1}$ and $v_{i_1},_{i_2},_{\cdots},_{i_{n-1}},_{i_n}=(i_1-1)N_2N_3\cdots N_n+(i_2-1)N_3N_4\cdots N_n+\cdots +(i_{n-1}-1)N_n+i_n$.
The corresponding sub-adjacency matrix is
\begin{eqnarray}
A_{i_1i_2\cdots i_{n-2},j_1j_2\cdots j_{n-2}}=  \nonumber
\label{L}
\end{eqnarray}
\begin{eqnarray}
\left[ \begin{array}{ccccc}
           A_{i_1i_2\cdots i_{n-2}1},_{j_1j_2\cdots j_{n-2}1}&A_{i_1i_2\cdots i_{n-2}1},_{j_1j_2\cdots j_{n-2}2}&\cdots&A_{i_1i_2\cdots i_{n-2}1},_{j_1j_2\cdots j_{n-2}N_{n-1}}   \\
           A_{i_1i_2\cdots i_{n-2}2},_{j_1j_2\cdots j_{n-2}1}&A_{i_1i_2\cdots i_{n-2}2},_{j_1j_2\cdots j_{n-2}2}&\cdots&A_{i_1i_2\cdots i_{n-2}2},_{j_1j_2\cdots j_{n-2}N_{n-1}}  \\
           \vdots&\vdots&\ddots&\vdots&   \\
          A_{i_1i_2\cdots i_{n-2}N_{n-1}},_{j_1j_2\cdots j_{n-2}1}&A_{i_1i_2\cdots i_{n-2}N_{n-1}},_{j_1j_2\cdots j_{n-2}2}&\cdots&A_{i_1i_2\cdots i_{n-2}N_{n-1}},_{j_1j_2\cdots j_{n-2}N_{n-1}}    \\
           \end{array}
      \right ],
\label{M}
\end{eqnarray}
where $A_{{i_1}{i_2}{\cdots}{i_{n-1}},{j_1}{j_2}{\cdots}{j_{n-1}}}$ are matrices of size $N_n$ representing edges between $C_{{i_1},{i_2},{\cdots},{i_{n-1}}}$ and $C_{{j_1},{j_2},{\cdots},{j_{n-1}}}$. Note that $A^{T}_{{i_1}{i_2}{\cdots}{i_{n-1}},{j_1}{j_2}{\cdots}{j_{n-1}}}
=A_{{j_1}{j_2}{\cdots}{j_{n-1}},{i_1}{i_2}{\cdots}{i_{n-1}}}$ (there exists $k\in$ $\{1,2,\cdots,n-1 $\} such that $i_k\ne j_k$) and $A^{T}_{{i_1}{i_2}{\cdots}{i_{n-1}},{i_1}{i_2}{\cdots}{i_{n-1}}}
=A_{{i_1}{i_2}{\cdots}{i_{n-1}},{i_1}{i_2}{\cdots}{i_{n-1}}}$,
because the adjacency matrix is symmetric.

We now introduce the concept of $GTPT$ in multipartite systems.

{\bf \noindent Definition 8}.
 A graph theoretical partial transpose $(GTPT)$ on graph $G$ is an operation of $G$ replacing all existing edges $(v_{{i_1},{i_2},{\cdots},{i_n}},v_{{j_1},{j_2},{\cdots},{j_n}})i_1\neq j_1$ by $(v_{{j_1},{i_2},{\cdots},{i_n}},v_{{i_1},{j_2},{\cdots},{j_n}})$ and keeping all other edges unchanged.

This gives rise to a new simple graph $G'=(V(G'),E(G'))$ from $G=(V(G),E(G))$, where $V(G')=V(G)$ with the same labelling.
We have
\begin{equation}
A(G)^{T_A}=\left\{
             \begin{array}{lr}
             1,\ ((j_1-1)N_2\cdots N_n +(i_2-1)N_3\cdots N_n+\cdots+(i_{n-1}-1)N_n+i_n, \\(i_1-1)N_2\cdots N_n +(j_2-1)N_3\cdots N_n+\cdots+(j_{n-1}-1)N_n+j_n)\in E(G');  \\
             0, \ otherwise.
             \end{array}
\right.
\label{N}
\end{equation}
Then $A(G')=A(G)^{T_A}$ and $\mid E(G')\mid=\mid E(G)\mid$.

We can similarly define degree symmetry in multipartite systems.

{\bf \noindent Definition 9}. A graph $G$ under a GTPT is called degree symmetric if $d_G(u)=d_G(u')$ for $u\in V(G)=V(G')$.

Using the method similar to Theorems 1, 
we have

\begin{thm} If graph $G$ under a GTPT is degree symmetric then full separability of $\rho_l(G)$ implies full separability of $\rho_l(G')$.\label{3}
\end{thm}

\section{ Partially Symmetric Graphs and Separability }

In this section, we study partially symmetric graphs and the relationship between partially symmetric and degree symmetric. We will prove
that a class of partially symmetric graphs are fully separable.

We begin with partial symmetry in tripartite combinatorial quantum systems.

{\bf \noindent Definition 10}. A graph $G$ is partially symmetric if $(v_{i,j,k},v_{s,u,v}) \in E(G)$ then $(v_{s,j,k},v_{i,u,v})\\ \in E(G)$ for all $i\ne s$ and $j,k,s,u,v$.

Note that if $G$ is partially symmetric, then its adjacency matrix satisfies that $A^T_{ij},_{kl}=A_{ij},_{kl}$ for all $i,j,k, l$.

If the graph $G$ is also under GTPT, then one can talk about degree symmetry and partial symmetry together. The following result explains their relationship.

\begin{lemma} Every partially symmetric graph $G$ is degree symmetric.
\label{(1)}
\end{lemma}
\begin{proof}Let $G$ be a partial symmetric graph, and graph $G'$ obtained from $G$ by GTPT. Clearly $V(G)=V(G')$. For any $v_{i,j,k}\in V(G)=V(G')$, if $(v_{i,j,k},v_{s,u,v})\in E(G)$ then
$(v_{s,j,k},v_{i,u,v})\in E(G)$ by partial symmetry of $G$.
But then $(v_{i,j,k},v_{s,u,v})\in E(G')$ by GTPT construction (applied to the second edge). On the other hand, for $(v_{i,j,k},v_{s,u,v})\in E(G')$. By construction $(v_{s,j,k},v_{i,u,v})\in E(G)$, which then
implies that $(v_{i,j,k},v_{s,u,v})\in E(G)$ by partial symmetry. Therefore $d_{G}(v_{i,j,k})=d_{G'}(v_{i,j,k})$, i.e. graph $G$ is degree symmetric.
\end{proof}
Next we provide a class of tripartite fully separable states of dimension $mnq$ arising from partial symmetric graphs.

\begin{thm}
Let $G$ be a partially symmetric graph with the following conditions:

\noindent(1) For any two vertices of any partition $C_i$, there is no edge, $(v_{i,j,k},v_{i,l,v})\notin E(G)$ for all $i,j,k,l\ and\ v$;

\noindent(2) For each layer, the following conditions hold:
\begin{itemize}
\item Either there is no edge between vertices of $C_i$ and $C_j$, or $A_{i,k}=A_{j,l}$ for all $i,j,k,l,i\ne \\ k$ and $ j\ne l$;

\item Either there is no edge between vertices of $C_{i,j}$ and $C_{s,t}$, or $A_{ij,kl}=A_{st,uv}$ for all $i,j,k,l,s,t,i\ne k\ or\ j\ne l$, and $s\ne u\ or\ t\ne v$;
\end{itemize}
\noindent(3) Degree of all vertices in a layer are the same.\\
Then $\rho_q(G)$ is fully separable.
\label{(5)}
\end{thm}
\begin{proof} As $G$ is partial symmetric, there is no edge between two vertices of any partition $C_i$. The normalized Laplacian reads that
\begin{eqnarray}
Q(G)=
\left[ \begin{array}{ccccccccccc}
          d_1I_q&0&\cdots&0&A_{11},_{21}&\cdots&A_{11},_{2n}&\cdots&A_{11},_{m1}&\cdots&A_{11},_{mn}   \\
          0&d_1I_q&\cdots&0&A_{12},_{21}&\cdots&A_{12},_{2n}&\cdots&A_{12},_{m1}&\cdots&A_{12},_{mn}   \\
          \vdots&\vdots&\ddots&\vdots&\vdots&\ddots&\vdots&\ddots&\vdots&\ddots&\vdots   \\
          0&0&\cdots&d_1I_q&A_{1n},_{21}&\cdots&A_{1n},_{2n}&\cdots&A_{1n},_{m1}&\cdots&A_{1n},_{mn}   \\
          A_{11},_{21}&A_{12},_{21}&\cdots&A_{1n},_{21}&d_2I_q&\cdots&0&\cdots&A_{21},_{m1}&\cdots&A_{21},_{mn}\\
            \vdots&\vdots&\ddots&\vdots&\vdots&\ddots&\vdots&\ddots&\vdots&\ddots&\vdots   \\
           A_{11},_{2n}&A_{12},_{2n}&\cdots&A_{1n},_{2n}&0&\cdots&d_2I_q&\cdots&A_{2n},_{m1}&\cdots&A_{2n},_{mn}\\
           \vdots&\vdots&\ddots&\vdots&\vdots&\ddots&\vdots&\ddots&\vdots&\ddots&\vdots   \\
           A_{11},_{m1}&A_{12},_{m1}&\cdots&A_{1n},_{m1}&A_{21},_{m1}&\cdots&A_{2n},_{m1}&\cdots&d_mI_q&\cdots&0\\
           \vdots&\vdots&\ddots&\vdots&\vdots&\ddots&\vdots&\ddots&\vdots&\ddots&\vdots   \\
           A_{11},_{mn}&A_{12},_{mn}&\cdots&A_{1n},_{mn}&A_{21},_{mn}&\cdots&A_{2n},_{mn}&\cdots&0&\cdots&d_mI_q\\
           \end{array}
      \right ],
\label{O}
\end{eqnarray}
where $I_q$ is the identity matrix of size $q$.

Write the symmetric (subadjacency) matrix $A_{zj,kl}=(a_{z,j,u},_{k,l,v})_{q\times q}$. Let $A_{zj,kl}=\sum_{r_1}\lambda_{r_1}u_{r_1}u^T_{r_1}$ be the spectral decomposition, where $u_{r_1}$ runs through a complete set of orthonormal eigenvectors associated with eigenvalues $\lambda_{r_1}, 1\leq r_1\leq q$. If there are no edges between vertices of $C_{z,j}$ and $C_{s,t}$, then $A_{zj,kl}=0$.
Otherwise $A_{zj,kl}=A_{st,uv}=\sum_{r_1}\lambda_{r_1}u_{r_1}u^T_{r_1}$.

Note that $Q(G)$ can also be written as
\begin{align} \label{AB}
Q(G)&=\sum_{r_1}
\left[ \begin{array}{ccccc}
           d_1I_n&H^{(\lambda_{r_1})}&H^{(\lambda_{r_1})}&\cdots&H^{(\lambda_{r_1})}   \\
           H^{(\lambda_{r_1})}&d_2I_n&H^{(\lambda_{r_1})}&\cdots&H^{(\lambda_{r_1})}   \\
           H^{(\lambda_{r_1})}&H^{(\lambda_{r_1})}&d_3I_n&\cdots&H^{(\lambda_{r_1})}\\
          \vdots&\vdots&\vdots&\ddots&\vdots   \\
          H^{(\lambda_{r_1})}&H^{(\lambda_{r_1})}&H^{(\lambda_{r_1})}&\cdots&d_mI_n\\
          \end{array}
      \right ]\otimes u_{r_1}u^T_{r_1}\\ \nonumber
&=\sum_{r_1}B^{(1)}\otimes u_{r_1}u^T_{r_1}
\end{align}
 where $B^{(1)}$ denotes the block matrix (the first factor) in the tensor decomposition and
 $H^{(\lambda_{r_1})}=(h^{(\lambda_{r_1})}_{ij})_{n\times n}$ are square matrices of size $n$. Here $h^{(\lambda_{r_1})}_{ij}=0$ or $\lambda_{r_1}$.

We claim that $B^{(1)}$ is a diagonally dominant matrix. Suppose a square matrix $A$ has eigenvalues $\lambda_i$,
the spectral radius $spr(A)$ is defined to be $\max_i |\lambda_i|$. Therefore
\begin{align*}
|\lambda_{r_1}|\leq spr(A_{zj,kl})\leq\parallel A_{zj,kl}\parallel_\infty,
\end{align*}
where $\parallel A_{zj,kl}\parallel_\infty=\max_u \sum^q_{v=1}|a_{z,j,u},_{k,l,v}|$.
Write $B^{(1)}=(b^{(1)}_{st})_{mn\times mn}$, then for each $s$
\begin{align*}
\sum_{t\neq s}|b^{(1)}_{st}|\leq \sum_{k,l}|\lambda_{r_1}|\leq \sum_{k,l} \max_u \sum^q_{v=1}|a_{z,j,u},_{k,l,v}|,
\end{align*}
so $B^{(1)}$ is a diagonally dominant with positive diagonal entries. Hence $B^{(1)}$ is a positive semidefinite matrix.

Since $H^{(\lambda_{r_1})}$ and $d_zI_n$ are commuting symmetric matrices, they can be simultaneously diagonalized.
 If $H^{(\lambda_{r_1})}\ne 0$, one can write that
 \begin{align*}
 H^{(\lambda_{r_1})}=\sum_{r_2}\lambda_{r_2}\ u_{r_2}u^T_{r_2},\ \ \ \
 d_zI_n=d_z\sum_{r_2}u_{r_2}u^T_{r_2},
 \end{align*}
 where $u_{r_2}$ form a complete set of orthonormal eigenvectors of $H^{(\lambda_{r_1})}$ corresponding to eigenvalues $\lambda_{r_2}$, $r_2=1,2,\cdots,n$. Thus we can write that
\begin{eqnarray}
Q{(G)}=\sum_{r_1}\sum_{r_2}
\left[ \begin{array}{ccccc}
           d_1&\lambda_{r_2}&\lambda_{r_2}&\cdots&\lambda_{r_2}  \\
           \lambda_{r_2}&d_2&\lambda_{r_2}&\cdots&\lambda_{r_2}   \\
           \lambda_{r_2}&\lambda_{r_2}&d_3&\cdots&\lambda_{r_2}\\
          \vdots&\vdots&\vdots&\ddots&\vdots   \\
          \lambda_{r_2}&\lambda_{r_2}&\lambda_{r_2}&\cdots&d_m\\
          \end{array}
      \right ]\otimes u_{r_2}u^T_{r_2}\otimes u_{r_1}u^T_{r_1}.
\label{P}
\end{eqnarray}
Note that
$$|\lambda_{r_2}|\leq spr (H^
{(\lambda_{r_1})})\leq\parallel H^{(\lambda_{r_1})}\parallel_\infty=\max_i\sum^n_{j=1}| h^{(\lambda_{r_1})}_{ij}|.$$
 Since $h^{(\lambda_{r_1})}_{ij}=0\ or\ \lambda_{r_1}$, it follows that for $1\leq z\leq m$
 $$d_z\geq\sum_{k\neq z}\sum^n_j|h^{(\lambda_{r_1})}_{ij}|=\sum_{k\neq z}\max_i\sum^n_{j=1}|h^{(\lambda_{r_1})}_{ij}|.$$

Let
\begin{eqnarray}
B^{(2)}=
\left[ \begin{array}{ccccc}
           d_1&\lambda_{r_2}&\lambda_{r_2}&\cdots&\lambda_{r_2}  \\
           \lambda_{r_2}&d_2&\lambda_{r_2}&\cdots&\lambda_{r_2}   \\
           \lambda_{r_2}&\lambda_{r_2}&d_3&\cdots&\lambda_{r_2}\\
          \vdots&\vdots&\vdots&\ddots&\vdots   \\
          \lambda_{r_2}&\lambda_{r_2}&\lambda_{r_2}&\cdots&d_m\\
          \end{array}
      \right ]=(b^{(2)}_{uv})_{m\times m}.
\label{Q}
\end{eqnarray}
Clearly
$$\sum_{1\leq u\neq v\leq m}|b^{(2)}_{uv}|\leq\sum_{k\neq z}|\lambda_{r_2}| \leq\sum_{k\neq z}\max_i \sum^n_{j=1}|h^{(\lambda{r_1})}_{ij}|\leq d_z, $$
which implies that $B^{(2)}$ is a diagonally dominant matrix. Hence $B^{(2)}$ is a positive semidefinite matrix. We have
\begin{align*}
\rho_q(G)=\sum_{r_1,r_2}\frac{B^{(2)}}{\tr(Q(G))}\otimes u_{r_2}u^T_{r_2}\otimes u_{r_1}u^T_{r_1}=\sum_{r_1,r_2}\frac{1}{nq}\frac{B^{(2)}}{d_1+\cdots+d_m}\otimes u_{r_2}u^T_{r_2}\otimes u_{r_1}u^T_{r_1},
\end{align*}
therefore $\frac1{d_1+\cdots +d_m}B^{(2)}, u_{r_2}u^T_{r_2}$ and $u_{r_1}u^T_{r_1}$ are positive semidefinite matrices with unit trace. So $\frac1{d_1+\cdots +d_m}B^{(2)}$, $u_{r_2}u^T_{r_2}$ and $u_{r_1}u^T_{r_1}$ are density matrices. Hence $\rho_q(G)$ is fully separable.
\end{proof}

{\bf \noindent Remark}: By the method of {} [21], one can obtain a criterion for biseparability. In contrast, our criterion
further gives detailed information on how the separable density matrix is expressed as a convex sum of tensor products.
Moreover,  though [21] gives biseparability of $\rho_q(G)$ for the bipartition $AB|C$,
it is unclear if it is fully separable. For example, the method does not
imply if the subsystem $AB$ is separable. \\

The above results can be generalized to multipartite quantum systems. First we extend the notion of partially symmetric graphs to multipartite systems.

{\bf \noindent Definition 11}. A graph $G$ is partially symmetric if $(v_{{i_1},{i_2},\cdots,{i_n}},v_{{j_1},{j_2},\cdots,{j_n}}) \in E(G)$ implies $(v_{{j_1},{i_2},\cdots,{i_n}},v_{{i_1},{j_2},\cdots,{j_n}}) \in E(G)$ $(\forall\ i_1,i_2,\cdots,i_n,j_1,j_2,\cdots,j_n,\ and\ i_1\neq j_1)$.

Note that, GTPT keeps a partial symmetric graph unchanged. If graph $G$ is partially symmetric, then $A^{T}_{{i_1i_2\cdots i_{n-1}},{j_1j_2\cdots j_{n-1}}}=A_{{i_1i_2\cdots i_{n-1}},{j_1j_2\cdots j_{n-1}}}$ for all $i_1,i_2,\cdots,i_{n-1},j_1,j_2,\cdots,j_{n-1}$. Using the similar method as lemma $1$, we see that a partially symmetric graph $G$ is also degree symmetric.

Now we consider a class of multipartite fully separable state of dimension $N_1N_2\cdots N_n$ arising from partially symmetric graphs.

\begin{thm}
Let $G$ be a partially symmetric graph with the following properties.

\noindent(1) Between two vertices of any partition $C_{i_1}$, there is no edge, $(v_{i_1,i_2,\cdots,i_n},v_{i_1,j_2,\cdots,j_n})\notin E(G)$, for all $i_1,i_2,\cdots,i_n,\ j_2,\cdots,j_n$.

\noindent(2) For each layer, the following conditions hold:
\begin{itemize}
 \item Either there is no edge between vertices of $C_{i_1}$ and $C_{j_1}$, or $A_{i_1,j_1}=A_{k_1,l_1}$ for all $i_1,k_1,j_1,l_1,i_1\ne  j_1\ and\ k_1\ne l_1$;

 \item Either there is no edge between vertices of $C_{i_1,i_2}$ and $C_{j_1,j_2}$, or $A_{i_1i_2,j_1j_2}=A_{k_1k_2,l_1l_2}$ for all $i_1,i_2,k_1,k_2,j_1,j_2,l_1,1_2,i_1\ne j_1\ or\ i_2\ne j_2$, and $k_1\ne l_1\ or\ k_2\ne l_2$;


 \item Either there is no edge between vertices of $C_{i_1,i_2,\cdots,i_{n-1}}$ and $C_{j_1,j_2,\cdots,j_{n-1}}$, or $A_{i_1i_2\cdots i_{n-1},j_1j_2\cdots j_{n-1}}$=\\$A_{k_1k_2\cdots k_{n-1},l_1l_2\cdots l_{n-1}}$ for all $i_1,i_2,\cdots,i_{n-1},k_1,k_2,\cdots,k_{n-1}, j_1,j_2,\cdots,j_{n-1},l_1,l_2,\cdots,l_{n-1}$, and there exist $g,h \in$ $\{1,2,\cdots,n $\} such that $i_g\neq j_g$ and $k_h\neq l_h$.
\end{itemize}
\noindent(3) Degree of all the vertices in a layer are the same.\\
Then $\rho_q(G)$ is fully separable.
\label{(6)}
\end{thm}
\begin{proof}
Let
\begin{eqnarray}
C^{(z)}_t=
\left[ \begin{array}{ccccc}
           d_tI_{N_{n-z}}&0&0&\cdots&0   \\
           0&d_tI_{N_{n-z}}&0&\cdots&0   \\
           0&0&d_tI_{N_{n-z}}&\cdots&0\\
          \vdots&\vdots&\vdots&\ddots&\vdots   \\
          0&0&0&\cdots&d_tI_{N_{n-z}}\\
          \end{array}
      \right ],
\label{R}
\end{eqnarray}
\begin{eqnarray}
D^{(z)}=
\left[ \begin{array}{ccccc}
           H^{(\lambda_{r_z})}&H^{(\lambda_{r_z})}&H^{(\lambda_{r_z})}&\cdots&H^{(\lambda_{r_z})}   \\
           H^{(\lambda_{r_z})}&H^{(\lambda_{r_z})}&H^{(\lambda_{r_z})}&\cdots&H^{(\lambda_{r_z})}   \\
           H^{(\lambda_{r_z})}&H^{(\lambda_{r_z})}&H^{(\lambda_{r_z})}&\cdots&H^{(\lambda_{r_z})}\\
          \vdots&\vdots&\vdots&\ddots&\vdots   \\
          H^{(\lambda_{r_z})}&H^{(\lambda_{r_z})}&H^{(\lambda_{r_z})}&\cdots&H^{(\lambda_{r_z})}\\
          \end{array}
      \right ],
\label{S}
\end{eqnarray}
\begin{eqnarray}
H^{(\lambda_{r_z})}=
\left[ \begin{array}{cccc}
           \lambda_{r_z}&\lambda_{r_z}&\cdots&\lambda_{r_z}   \\
           \lambda_{r_z}&\lambda_{r_z}&\cdots&\lambda_{r_z}\\
          \vdots&\vdots&\ddots\vdots \\
          \lambda_{r_z}&\lambda_{r_z}&\cdots&\lambda_{r_z}   \\
          \end{array}
      \right ]=(h^{(\lambda_{r_z})}_{ij})_{N_{n-z}\times N_{n-z}},
\label{T}
\end{eqnarray}
and
\begin{eqnarray}
B^{(z')}=
\left[ \begin{array}{ccccc}
           C^{(z')}_1&D^{(z')}&D^{(z')}&\cdots&D^{(z')}   \\
           D^{(z')}&C^{(z')}_2&D^{(z')}&\cdots&D^{(z')}   \\
           D^{(z')}&D^{(z')}&C^{(z')}_3&\cdots&D^{(z')}\\
          \vdots&\vdots&\vdots&\ddots&\vdots   \\
         D^{(z')}&D^{(z')}&D^{(z')}&\cdots&C^{(z')}_{N_1}\\
          \end{array}
      \right ] ,
\label{U}
\end{eqnarray}
where $B^{(z')}$ is a square matrix of size $ N_1N_2\cdots N_{n-z'}$, $z=1,2,\cdots,n-3, z'=1,2,\cdots,n-1$.

Let $A_{i_1i_2\cdots i_{n-1},j_1j_2\cdots j_{n-1}}=(a_{{i_1,i_2,\cdots, i_n},{j_1,j_2,\cdots, j_n}})_{N_n\times N_n}$ be a symmetric matrix, and suppose its spectral decomposition is given by
$$A_{i_1i_2\cdots i_{n-1},j_1j_2\cdots j_{n-1}}=\sum_{r_1}\lambda_{r_1}u_{r_1}u^T_{r_1},(r_1=1,2,\cdots,N_n)£¬$$
where $u_{r_1}$ form a complete set of orthonormal eigenvectors corresponding to the eigenvalues $\lambda_{r_1}$. If there is no edge between vertices of $C_{i_1i_2\cdots i_{n-1}}$ and $C_{j_1j_2\cdots j_{n-1}}$, then $A_{i_1i_2\cdots i_{n-1},j_1j_2\cdots j_{n-1}}=\sum_{r_1}0u_{r_1}u^T_{r_1}$. Subsequently $A_{i_1i_2\cdots i_{n-1},j_1j_2\cdots j_{n-1}}=A_{k_1k_2\cdots k_{n-1},l_1l_2\cdots l_{n-1}}=\sum_{r_1}\lambda_{r_1}u_{r_1}u^T_{r_1}$.

$Q(G)$ can also be written as
\begin{eqnarray}
Q(G)=\sum_{r_1}
\left[ \begin{array}{ccccc}
           C^{(1)}_1&D^{(1)}&D^{(1)}&\cdots&D^{(1)}   \\
           D^{(1)}&C^{(1)}_2&D^{(1)}&\cdots&D^{(1)}   \\
           D^{(1)}&D^{(1)}&C^{(1)}_3&\cdots&D^{(1)}\\
          \vdots&\vdots&\vdots&\ddots&\vdots   \\
         D^{(1)}&D^{(1)}&D^{(1)}&\cdots&C^{(1)}_{N_1}\\
          \end{array}
      \right ]\otimes u_{r_1}u^T_{r_1},
\label{V}
\end{eqnarray}
where $C^{(1)}_t (t=1,2,\cdots,N_1)$ and $D^{(1)}$ are squre matrices of size $N_2N_3\cdots N_{n-1}$. Note that
$H^{(\lambda_{r_1})}$ is a square matrix
of size $N_{n-1}$, and $h^{(\lambda_{r_1})}_{ij}=0\ or\ \lambda_{r_1}$.

Let $B^{(1)}=(b^{(1)}_{uv})$. Using the similar way of Theorem 3, we get that
$$\sum_{u\neq v}|b^{(1)}_{uv}|\leq \sum_{j_1,j_2,\cdots,j_{n-1}}|\lambda_{r_1}|\leq \sum_{j_1,j_2,\cdots,j_{n-1}} \max_{i_n} \sum^{N_n}_{j_n=1}|a_{i_1,i_2,\cdots,i_n},_{j_1,j_2,\cdots,j_n}|, $$
so $B^{(1)}$ is a diagonally dominant matrix. Hence $B^{(1)}$ is a positive semidefinite matrix.

 Since $H^{(\lambda_{r_1})}$ and $d_tI_{N_{n-1}}$ are commuting matrices, there exists a common set of eigenvectors such that
 \begin{align*}
 H^{(\lambda_{r_2})}&=\sum_{r_2}\lambda_{r_2}\ u_{r_2}u^T_{r_2}\\
 d_tI_{N_{n-1}}&=d_t\sum_{r_2}u_{r_2}u^T_{r_2}
 \end{align*}
 then
\begin{eqnarray}
Q{(G)}=\sum_{r_1}\sum_{r_2}
\left[ \begin{array}{ccccc}
           C^{(2)}_1&D^{(2)}&D^{(2)}&\cdots&D^{(2)}   \\
           D^{(2)}&C^{(2)}_2&D^{(2)}&\cdots&D^{(2)}   \\
           D^{(2)}&D^{(2)}&C^{(2)}_3&\cdots&D^{(2)}\\
          \vdots&\vdots&\vdots&\ddots&\vdots   \\
         D^{(2)}&D^{(2)}&D^{(2)}&\cdots&C^{(2)}_{N_1}\\
          \end{array}
      \right ]\otimes u_{r_2}u^T_{r_2}\otimes u_{r_1}u^T_{r_1},
\label{W}
\end{eqnarray}
where $C^{(2)}_t (t=1,2,\cdots,N_1)$ and $D^{(2)}$ are matrices of size $N_2N_3\cdots N_{n-2}$. Note that
$H^{(\lambda_{r_2})}$ is a matrix of size $N_{n-2}$, and $h^{(\lambda_{r_1})}_{ij}=0\ or\ \lambda_{r_1}$.

Next we show that $B^{(2)}$ is a diagonally dominant matrix. In fact,
$$|\lambda_{r_2}|\leq spr H^
{(\lambda_{r_1})}\leq\parallel H^{(\lambda_{r_1})}\parallel_\infty=\max_i\sum^{N_{n-1}}_{j=1}| h^{(\lambda_{r_1})}_{ij}|.$$
Since $h^{(\lambda_{r_1})}_{ij}=0\ or\ \lambda_{r_1}$,
$$d_t\geq\sum_{i_1,i_2,\cdots,i_{n-1}\neq t}\sum^{N_{n-1}}_j|h^{(\lambda_{r_1})}_{ij}|=\sum_{i_1,i_2,\cdots,i_{n-1}\neq t}\max_i\sum^{N_{n-1}}_{j=1}|h^{(\lambda_{r_1})}_{ij}|.$$
Write $B^{(2)}=(b^{(2)}_{uv})$, where $u, v=1,\ldots, N_1N_2\cdots N_{n-2}$, then for $t=1,2,\cdots,N_1$
$$\sum_{u\neq v}|b^{(2)}_{uv}|\leq\sum_{i_1,i_2,\cdots,i_{n-1}\neq t}|\lambda_{r_2}| \leq\sum_{i_1,i_2,\cdots,i_{n-1}\neq t}\max_i \sum^{N_{n-1}}_{j=1}|h^{(\lambda{r_1})}_{ij}|\leq d_t .$$
Therefore $B^{(2)}$ is a diagonally dominant matrix, and subsequently a positive semi-definite matrix.

Thus we can write that
\begin{eqnarray}
Q{(G)}=\sum_{r_1,\ldots, r_{n-2}}
\left[ \begin{array}{ccccc}
           C^{(n-2)}_1&D^{(n-2)}&D^{(n-2)}&\cdots&D^{(n-2)}   \\
           D^{(n-2)}&C^{(n-2)}_2&D^{(n-2)}&\cdots&D^{(n-2)}   \\
           D^{(n-2)}&D^{(n-2)}&C^{(n-2)}_3&\cdots&D^{(n-2)}\\
          \vdots&\vdots&\vdots&\ddots&\vdots   \\
         D^{(n-2)}&D^{(n-2)}&D^{(n-2)}&\cdots&C^{(n-2)}_{N_1}\\
          \end{array}
      \right ]\otimes u_{r_{n-2}}u^T_{r_{n-2}}\otimes\cdots \otimes u_{r_1}u^T_{r_1},
\label{X}
\end{eqnarray}
where
\begin{eqnarray}
C^{(n-2)}_t=
\left[ \begin{array}{cccc}
           d_t&0&\cdots&0   \\
           0&d_t&\cdots&0   \\
          \vdots&\vdots&\ddots&\vdots   \\
          0&0&\cdots&d_t\\
          \end{array}
      \right ],
\label{Y}
\end{eqnarray}

\begin{eqnarray}
D^{(n-2)}=H^{(\lambda_{r_{n-2}})}=
\left[ \begin{array}{cccc}
           \lambda_{r_{n-2}}&\lambda_{r_{n-2}}&\cdots&\lambda_{r_{n-2}}   \\
           \lambda_{r_{n-2}}&\lambda_{r_{n-2}}&\cdots&\lambda_{r_{n-2}}\\
          \vdots&\vdots&\ddots&\vdots \\
          \lambda_{r_{n-2}}&\lambda_{r_{n-2}}&\cdots&\lambda_{r_{n-2}}   \\
          \end{array}
          \right ]=(h^{(\lambda_{r_{n-2}})}_{ij})_{N_2\times N_2},
\label{Z}
\end{eqnarray}
where $C^{(n-2)}_t (t=1,2,\cdots,N_1)$, $D^{(n-2)}$ and $H^{(\lambda_{r_{n-2}})}$ are square matrices of size $N_2$. Here $h^{(\lambda_{r_{n-2}})}_{ij}=0\ or\ \lambda_{r_{n-2}}$.

Since $H^{(\lambda_{r_{n-2}})}$ and $d_tI_{N_2}$ are commuting symmetric matrices, they have common eigenvectors: $u_{r_{n-1}}$.
Therefore we can write that
\begin{eqnarray}
Q{(G)}=\sum_{r_1}\cdots \sum_{r_{n-1}}
\left[ \begin{array}{ccccc}
           d_1&\lambda_{r_{n-1}}&\lambda_{r_{n-1}}&\cdots&\lambda_{r_{n-1}}  \\
           \lambda_{r_{n-1}}&d_2&\lambda_{r_{n-1}}&\cdots&\lambda_{r_{n-1}}   \\
           \lambda_{r_{n-1}}&\lambda_{r_{n-1}}&d_3&\cdots&\lambda_{r_{n-1}}\\
          \vdots&\vdots&\vdots&\ddots&\vdots   \\
          \lambda_{r_{n-1}}&\lambda_{r_{n-1}}&\lambda_{r_2}&\cdots&d_{N_1}\\
          \end{array}
      \right ]\otimes u_{r_{n-1}}u^T_{r_{n-1}}\cdots \otimes u_{r_1}u^T_{r_1},
\label{AA}
\end{eqnarray}
Let $B^{(n-1)}=(b^{(n-1)}_{uv})_{N_1\times N_1}$, then
$$\sum_{u\neq v}|b^{(n-1)}_{uv}|\leq \sum_{i_1\neq t}|\lambda_{r_{n-1}}|\leq\sum_{i_1\neq t}\sum^{N_2}_{j=1}|h^{\lambda_{r_{n-2}}}_{ij}|=\sum_{i_1\neq t}\max_i\sum^{N_2}_{j=1}|h^{\lambda_{r_{n-2}}}_{ij}|\leq d_t.$$
So $B^{(n-1)}$ is a diagonally dominant matrix, thus positive semidefinite matrix.
Putting all these together, we have that for $r_k=N_{n-k+1},k=1,2,\cdots,n-1$
\begin{align*}
\rho_q(G)&=\frac{\sum_{r_1,\cdots,r_{n-1}}B^{(n-1)}\otimes u_{r_{n-1}}u^T_{r_{n-1}}\otimes\cdots \otimes u_{r_1}u^T_{r_1}}{\tr(Q(G))}\\
&=\sum_{r_1,\cdots,r_{n-1}}\frac{B^{(n-1)}}{\tr(Q(G))}\otimes u_{r_{n-1}}u^T_{r_{n-1}}\otimes\cdots\otimes u_{r_1}u^T_{r_1}\\
&=\sum_{r_1,\cdots,r_{n-1}}\frac{1}{N_2\cdots N_n}\frac{B^{(n-1)}}{d_1+d_2+\cdots+d_{N_1}}\otimes u_{r_{n-1}}u^T_{r_{n-1}}\otimes\cdots \otimes u_{r_1}u^T_{r_1}
\end{align*}
therefore $\frac{B^{(n-1)}}{d_1+d_2+\cdots+d_{N_1}}, u_{r_2}u^T_{r_2},\cdots,\ u_{r_1}u^T_{r_1}$
are positive semidefinite matrices with unit trace. So $\frac{B^{(n-1)}}{d_1+d_2+\cdots+d_{N_1}}, u_{r_2}u^T_{r_2},\cdots, u_{r_1}u^T_{r_1}$ are density matrices. Hence $\rho_q(G)$ is fully separable.
\end{proof}

\section{ Conclusion}

We have studied the separability for multipartite quantum states arising from simple graphs. With the properties of degree symmetric graphs, we have proved that separability of $\rho_l(G)$ implies that of $\rho_l(G')$ for the degree symmetric graph $G$ in the tripartite systems. And these results have been generalized to multipartite systems. Furthermore, we have studied the properties of partially symmetric graphs and proved that every partially symmetric graph $G$ is degree symmetric. We have presented a new layers method and provided classes of tripartite and multipartite fully separable states arising from partially symmetric graphs. These results are useful to distinguish separable states. It is hoped that
this work may help understand the physical characteristics and mathematical structures of (graph) separable states.

\noindent\textbf{Acknowledgments} This work is supported by the National Natural Science Foundation of China under grant Nos. 11101017,
11531004 and 11726016 and Simons Foundation grant No. 523868. 

\vspace*{2mm}
\begin{center}\raggedright{\bf References}\end{center}
\par\parindent 0.5em\hangafter 1 \hangindent 1.9em
{\bf 1.} Ekert, A.K.: Quantum cryptogaphy based on Bell's theorem. Phys. Rev. Lett. {\bf 67}(6), 661-663 (1991).

\par\parindent 0.5em\hangafter 1 \hangindent 1.9em
{\bf 2.} Bennett, C.H., Wiesner, S.J.: Communication via one- and two-particle operators on Einstein-Podolsky-Rosen states. Phys. Rev. Lett. {\bf 69}(20), 2881-2884 (1992).

\par\parindent 0.5em\hangafter 1 \hangindent 1.9em
{\bf 3.} Bennett, C.H., Brassard, G., Cr\'epeau, C., Jozsa, R., Peres, A.,Wootters, W.K.: Teleporting an unknown quantum state via dual classical and Einstein-Podolsky-Rosen channels. Phys. Rev. Lett. {\bf 70}(13), 1895-1899 (1993).

\par\parindent 0.5em\hangafter 1 \hangindent 1.9em
{\bf 4.} Bapat, R.B.: Graphs and matrices. Springer-Verlag London, (2010).

\par\parindent 0.5em\hangafter 1 \hangindent 1.9em
{\bf 5.} Braunstein, S.L.£¬Ghosh, S., Severini, S.: The Laplacian of a graph as a density matrix: a basic combinatorial approach to separability of mixed states. Ann. Combin. {\bf 10}(3), 291-317 (2006).  

\par\parindent 0.5em\hangafter 1 \hangindent 1.9em
{\bf 6.} Hein, M., Eisert, J., Briegel, H.J.: Multiparty entanglement in graph states. Phys. Rev. A. {\bf 69}, 062311 (2004). 

\par\parindent 0.5em\hangafter 1 \hangindent 1.9em
{\bf 7.} Anders, S., Briegel, H.J.: Fast simulation of stabilizer circuits using a graph-state representation. Phys. Rev. A. {\bf 73}, 022334 (2006).

\par\parindent 0.5em\hangafter 1 \hangindent 1.9em
{\bf 8.} Benjamin, S.C., Browne, D.E., Fitzsimons, J., Morton, J.J.L.: Brokered graph-state quantum computation. New Journal of Physics {\bf 8}, 141 (2006).

\par\parindent 0.5em\hangafter 1 \hangindent 1.9em
{\bf 9.} Singh, S.K., Pal, S.P.,  Kumar, S., Srikanth, R.: A combinatorial approach for studying local operations and classical communication transformations of multipartite states. J. Math. Phys. {\bf 46}, 122105 (2005).

\par\parindent 0.5em\hangafter 1 \hangindent 1.9em
{\bf 10.} Dutta, S., {Adhikari, B.,} Banerjee, S.: A graph theoretical approach to states and unitary operations. Quantum Inf. Process. {\bf 15}(5), 2193-2212 (2016). 

\par\parindent 0.5em\hangafter 1 \hangindent 1.9em
{\bf 11.} Wu, C. W.: 
Conditions for separability in generlized Laplacian matrices and diagonally dominant matrices as density matrices. Phys. Lett. A {\bf 351}(1), 18-22 (2006)

\par\parindent 0.5em\hangafter 1 \hangindent 1.9em
{\bf 12.}  Wu, C. W.: Multipartite separability of Laplacian matrices of graphs. {The electronic journal of combinatorics} {\bf 16}(1), R61 (2009).

\par\parindent 0.5em\hangafter 1 \hangindent 1.9em
{\bf 13.}  Wu, C. W.: On graphs whose Laplacian matrix's multipartite separability is invariant under graph isomorphism. Discrete Math. {\bf 310}(21), 2811-2814 (2010).

\par\parindent 0.5em\hangafter 1 \hangindent 1.9em
{\bf 14.}  Wu, C. W.: Graphs whose normalized Laplacian matrices are separable as density matrices in quantum mechanics. Discrete Math. {\bf 339}(4), 1377-1381 (2016).

\par\parindent 0.5em\hangafter 1 \hangindent 1.9em
{\bf 15.} Zhao, H., Fan, J.: Separability of generalized graph product states. Chin. Phy. Lett {\bf 30}(9), 090303 (2013).

\par\parindent 0.5em\hangafter 1 \hangindent 1.9em
{\bf 16.} Hildebrand, R., Mancini, S., Severini, S.:
Combinatorial laplacians and positivity under partial transpose. Math. Struct. in Comp. Sci. {\bf 18}(1), 205-219 (2008).

\par\parindent 0.5em\hangafter 1 \hangindent 1.9em
{\bf 17.} Braunstein, S.L., Ghosh, S., Mansour, T., Severini, S., Wilson, R. C.:
Some families of density matrices for which separability is easily tested. Phys. Rev. A. {\bf 73}, 012320 (2006)

\par\parindent 0.5em\hangafter 1 \hangindent 1.9em
{\bf 18.} Xie, C., Zhao, H., Wang, Z.X.: Separability of density matrices of graphs for multipartite systems. Electr. J. Combin. {\bf 20}(4), P21 (2013).

\par\parindent 0.5em\hangafter 1 \hangindent 1.9em
{\bf 19.} Dutta, S., Adhikari, B., Banerjee, S., Srikanth, R.: Bipartite separability and nonlocal quantum operations on graphs. Phys. Rev. A {\bf 94}, 012306 (2016).

\par\parindent 0.5em\hangafter 1 \hangindent 1.9em
{\bf 20.} Ha, K.-C., Kye, S.-H.: Construction of three-qubit genuine entanglement with bipartite positive partial transposes. Phys. Rev. A {\bf 93}, 032315 (2016).

\par\parindent 0.5em\hangafter 1 \hangindent 1.9em 
{\bf 21.} Ha, K.-C.: Sufficient criterion for separability of bipartite states. Phys. Rev. A {\bf 82}, 032315 (2010).

\end{document}